%% file: skipping-refinement.tex
\pdfoutput=1
\documentclass{llncs}
\usepackage{paralist,array,amsmath,amsfonts,amssymb,stmaryrd,hyperref,relsize,extarrows,centernot,esvect}
\usepackage{color,booktabs,fixltx2e,tikz-cd,graphicx,subfig,etex,etoolbox,textcomp}

\usepackage{gnuplot-lua-tikz}

\newtheorem{thm}{Theorem}
\newtheorem{lem}[thm]{Lemma}
\newcommand{\ctree}{\ensuremath{\mathit{ctree}}}
\newcommand{\ranktct}{\ensuremath{\mathit{ranktCt}}}
\newcommand{\size}{\ensuremath{\mathit{size}}}
\newcommand\mmit[1]{\ensuremath{\mathit{#1}}}
\newcommand{\La}[0]{\ensuremath{\langle}}
\newcommand{\Ra}[0]{\ensuremath{\rangle}}

\newcommand{\rankt} {\ensuremath{\mathit{rankt}}}
\newcommand{\rankl} {\ensuremath{\mathit{rankl}}}

\newcommand\corr[5]{\ensuremath{\mathit{corr(#1,#2,#3,#4,#5)}}}
\newcommand\match[3]{\ensuremath{\mathit{match(#1,#2,#3)}}}
\newcommand{\partition}[3]{\ensuremath{\mathbin{^{#2}#1^{#3}}}}

\newcommand{\ie}[0]{\emph{i.e.}, }
\newcommand{\eg}[0]{\emph{e.g.}, }
\def\myiff{\ensuremath{\mathit{iff}}}

\newcommand{\trans}[0]{\ensuremath{ \rightarrow }}
\newcommand{\transplus}[0]{\ensuremath{ \rightarrow^{+} }}

\newcommand{\trs}[1]{\mbox{\ensuremath{\mathcal{M_{\text{#1}}}= \langle S_{#1}, \xrightarrow{#1}, L_{#1} \rangle}}}
\newcommand{\labf}{\ensuremath{\mathcal{L}}}
\newcommand{\M}[1]{\ensuremath{\mathcal{M}_{#1}}}

\newcommand{\disjtrs}[3]{\ensuremath{\langle S_{#1} \uplus S_{#2},
\xrightarrow{#1} \uplus \xrightarrow{#2}, \labf{#3} \rangle }}

\newcommand{\fp}{\textit{fp}}
\newcommand{\inc}{\textit{INC}}

\let\orgautoref\autoref %

\renewcommand{\autoref}[1] {%
  \def\equationautorefname{Eq.}%
  \def\figureautorefname{Fig.}%
  \def\subfigureautorefname{Fig.}%
  \def\subfigureautorefname{Fig.}%
  \def\definitionautorefname{Definition}%
  \orgautoref{#1}%
}

\newcommand{\from}{\ensuremath{\colon}}
\newcommand{\ra}{\ensuremath{\rightarrow\;}}

\DeclareFontFamily{OT1}{bbm}{}
\DeclareFontShape{OT1}{bbm}{m}{n}{%
      <5> <6> <7> <8> <9> <10> <12> <17> gen * bbm%
      <10-12>bbm10%
      <12-17>bbm12%
      <17->bbm17%
      <-5>bbm5}{}

\DeclareSymbolFont{bm}{OT1}{bbm}{m}{n}
\DeclareMathSymbol{\N}{7}{bm}{'116}

\newcommand{\ffp}{\hspace{.05in}\ensuremath{\square}}

\begin{document}
\title{Skipping Refinement}
 \author{Mitesh Jain \and Panagiotis Manolios}
 \institute{Northeastern University\thanks{This research was supported in part by DARPA under AFRL
Cooperative Agreement No.~FA8750-10-2-0233 and by NSF grants
CCF-1117184 and CCF-1319580.}\\
  \email{\{jmitesh,pete\}@ccs.neu.edu}}
\maketitle

\begin{abstract}
  We introduce skipping refinement, a new notion of
  correctness for reasoning about optimized reactive systems.
  Reasoning about reactive systems using refinement involves defining
  an abstract, high-level \emph{specification} system and a concrete,
  low-level \emph{implementation} system. One then shows that every
  behavior allowed by the implementation is also allowed by the
  specification. Due to the difference in abstraction levels, it is
  often the case that the implementation requires many steps to match
  one step of the specification, hence, it is quite useful for
  refinement to directly account for \emph{stuttering}.  Some
  optimized implementations, however, can actually take multiple
  specification steps at once.  For example, a memory controller can
  buffer the commands to the memory and at a later time simultaneously
  update multiple memory locations, thereby \emph{skipping} several
  observable states of the abstract specification, which only updates
  one memory location at a time. We introduce skipping simulation
  refinement and provide a sound and complete characterization
  consisting of ``local'' proof rules that are amenable to
  mechanization and automated verification. We present case studies
  that highlight the applicability of skipping refinement: a
  JVM-inspired stack machine, a simple memory controller and a
  scalar to vector compiler transformation. Our experimental results
  demonstrate that current model-checking and automated theorem
  proving tools have difficultly automatically analyzing these systems
  using existing notions of correctness, but they can analyze the
  systems if we use skipping refinement.
\end{abstract} 

\section{Introduction}

Refinement is a powerful method for reasoning about reactive
systems. The idea is to prove that every execution of the
concrete system being verified is allowed by the abstract
system. The concrete system is defined at a lower level of
abstraction, so it is usually the case that it requires several
steps to match one high-level step of the abstract system. Thus,
notions of refinement usually directly account for
stuttering~\cite{browne1988characterizing,van1990linear,lamport1991existence}.

Engineering ingenuity and the drive to build ever more efficient
systems has led to highly-optimized concrete systems capable of
taking \emph{single} steps that perform the work of
\emph{multiple} abstract steps.  For example, in order to reduce
memory latency and effectively utilize memory bandwidth, memory
controllers often buffer requests to memory. The pending requests
in the buffer are analyzed for address locality and then at some
time in the future, multiple locations in the memory are read and
updated simultaneously. Similarly, to improve instruction
throughput, superscalar processors fetch multiple instructions in
a single cycle.  These instructions are analyzed for
instruction-level parallelism (\eg the absence of data
dependencies) and, where possible, are executed in parallel,
leading to multiple instructions being retired in a single
cycle. In both these examples, in addition to stuttering, a
single step in the implementation may perform the work of
multiple abstract steps, \eg by updating multiple locations in
memory and retiring multiple instructions in a single
cycle. Thus, notions of refinement that only account for
stuttering are not appropriate for reasoning about such optimized
systems. In Section~\ref{sec:sks}, we introduce \emph{skipping
  refinement}, a new notion of correctness for reasoning about
reactive systems that ``execute faster'' and therefore can skip
some steps of the specification.  Skipping can be thought of as
the dual of stuttering: stuttering allows us to ``stretch''
executions of the specification system and skipping allows us to
``squeeze'' them.

An appropriate notion of correctness is only part of the
story. We also want to leverage the notion of correctness in
order to mechanically verify systems. To this end, in
Section~\ref{sec:automated-reasoning}, we introduce \emph{Well-Founded
Skipping}, a sound and complete characterization of skipping
simulation that allows us to prove refinement theorems about the
kind of systems we consider using only local reasoning. This
characterization establishes that refinement maps always exist
for skipping refinement. In Section~\ref{sec:casestudies}, we
illustrate the applicability of skipping refinement by
mechanizing the proof of correctness of three systems: a stack
machine with an instruction buffer, a simple memory controller,
and a simple scalar-to-vector compiler transformation. We show
experimentally that by using skipping refinement current
model-checkers are able to verify systems that otherwise are
beyond their capability to verify. We end with related work and
conclusions in
Sections~\ref{sec:related}~and~\ref{sec:conclusions}.

Our contributions include (1) the introduction of skipping
refinement, which is the first notion of refinement to directly
support reasoning about optimized systems that execute faster
than their specifications (as far as we know) (2) a sound and
complete characterization of skipping refinement that requires
only local reasoning, thereby enabling automated verification and
showing that refinement maps always exist (3) experimental
evidence showing that the use of skipping refinement allows us to
extend the complexity of systems that can be automatically
verified using state-of-the-art model checking and interactive
theorem proving technology.

\section{Motivating Examples}

To illustrate the notion of skipping simulation, we consider a
running example of a discrete-time event simulation (DES) system.
A state of the abstract, high-level specification system is a
three-tuple $\langle t, E, A\rangle$ where $t$ is a natural
number corresponding to the current time, $E$ is a set of pairs
$(e, t_e)$ where $e$ is an event scheduled to be executed at time
$t_e$ (we require that $t_e \geq t$), and $A$ is an assignment of
values to a set of (global) state variables.  The transition
relation for the abstract DES system is defined as follows. If
there is no event of the form $(e, t) \in E$, then there is
nothing to do at time $t$ and so $t$ is incremented by
1. Otherwise, we (nondeterministically) choose and execute an
event of the form $(e, t) \in E$. The execution of an event can
modify the state variables and can also generate a finite number
of new events, with the restriction that the time of any
generated event is $> t$. Finally, execution involves removing
$(e, t)$ from $E$.

Now, consider an optimized, concrete implementation of the
abstract DES system. As before, a state is a three-tuple $\langle
t, E, A \rangle$. However, unlike the abstract system which
just increments time by 1 when no events are scheduled for the
current time, the optimized system uses a priority queue
to find the next event to execute.  The transition relation is
defined as follows. An event $(e, t_e)$ with the minimum time is
selected, $t$ is updated to $t_e$ and the event $e$ is executed,
as above.

Notice that the optimized implementation of the discrete-time
event simulation system can run faster than the abstract
specification system by \emph{skipping} over abstract states when
no events are scheduled for execution at the current time. This
is neither a stuttering step nor corresponds to a single step of the
specification. Therefore, it is not possible to prove that the
implementation refines the specification using notions of
refinement that only allow
stuttering~\cite{lamport1991existence,manolios2003compositional},
because that just is not true. But, intuitively, there is a sense
in which the optimized DES system \emph{does} refine the abstract
DES system. Skipping refinement is our attempt at formally
developing the theory required to rigorously reason about these
kinds of systems.

Due to its simplicity, we will use the discrete-time event
simulation example in later sections to illustrate various
concepts.  After the basic theory is developed, we provide an
experimental evaluation based on three other motivating examples.
The first is a JVM-inspired stack machine that can store
instructions in a queue and then process these instructions in
bulk at some later point in time. The second example is an
optimized memory controller that buffers requests to memory to
reduce memory latency and maximize memory bandwidth
utilization. The pending requests in the buffer are analyzed for
address locality and redundant writes and then at some time in
the future, multiple locations in the memory are read and updated
in a single step. The final example is a compiler transformation
that analyzes programs for superword-level parallelism
and, where possible, replaces multiple scalar instructions with a
compact SIMD instruction that concurrently operates on multiple
words of data.  All of these examples require skipping, because
the optimized concrete systems can do more than inject stuttering
steps in the executions specified by their specification systems;
they can also collapse executions.

\section{Skipping Simulation and Refinement}
\label{sec:sks}

In this section, we introduce the notions of skipping simulation
and refinement.  We do this in the general setting of labeled
transition systems where we allow state space sizes and branching
factors of arbitrary infinite cardinalities.

We start with some notational conventions. Function application
is sometimes denoted by an infix dot ``$.$'' and is
left-associative. For a binary relation $R$, we often write $xRy$
instead of $(x, y) \in R$.  The composition of relation $R$ with
itself $i$ times (for $0 < i \leq \omega$) is denoted $R^i$
($\omega = \N$ and is the first infinite ordinal).  Given a
relation $R$ and $1<k\leq \omega$, $R^{<k}$ denotes $\bigcup_{1
  \leq i < k} R^i$ and $R^{\geq k}$ denotes $\bigcup_{\omega > i
  \geq k} R^i$ .  Instead of $R^{< \omega}$ we often write the
more common $R^+$.  $\uplus$ denotes the disjoint union operator.
Quantified expressions are written as $\langle \emph{Q}x \from r
\from p \rangle$, where \emph{Q} is the quantifier (\eg $\exists,
\forall$), $x$ is the bound variable, $r$ is an expression that
denotes the range of \emph{x} (\emph{true} if omitted), and $p$
is the body of the quantifier.

\begin{definition}
  A labeled transition system (TS) is a structure \mbox{$\langle
    S, \trans,L\rangle$}, where $S$ is a non-empty (possibly
  infinite) set of states, $\trans \ \subseteq S \times S$ is a
  left-total transition relation (every state has a successor),
  and $L$ is the \emph{labeling function}: its domain is $S$ and
  it tells us what is observable at a state.

  A path is a sequence of states such that for adjacent states
  $s$ and $u$, $s \ra u$. A path, $\sigma$, is a \emph{fullpath}
  if it is infinite. $\fp.\sigma.s$ denotes that $\sigma$ is a
  fullpath starting at s and for $i \in \omega, \sigma(i)$
  denotes the $i^{th}$ element of path $\sigma$.
\end{definition}

Our definition of skipping simulation is based on the notion of
\emph{matching}, which we define below. Informally, we say a
fullpath $\sigma$ matches a fullpath $\delta$ under relation $B$
if the fullpaths can be partitioned into non-empty, finite
segments such that all elements in a particular segment of
$\sigma$ are related to the first element in the corresponding
segment of $\delta$.

 \begin{definition}[Match]
   \label{def:match}
   Let \inc\; be the set of strictly increasing sequences of natural
   numbers starting at 0. Given a fullpath $\sigma$, the $i^{th}$
   segment of $\sigma$ with respect to $\pi \in \inc$, written
   $\partition{\sigma}{\pi}{i}$, is given by the sequence
   \mbox{\La$\sigma(\pi.i),...., \sigma(\pi.(i+1) -1)\Ra$.}
   For $\pi, \xi \in \inc$ and relation $B$, we define
   {\setlength{\abovedisplayskip}{3pt}
     \setlength{\belowdisplayskip}{5pt}
     \begin{flalign*}
       &\corr{B}{\sigma}{\pi}{\delta}{\xi} \equiv \La \forall i \in
       \omega:: \La \forall s \in \partition{\sigma}{\pi}{i} ::
       s B \delta(\xi.i) \Ra \Ra \textit{ and } &\\  
       &\match{B}{\sigma}{\delta} \equiv \La \exists \pi, \xi \in
       \inc :: \corr{B}{\sigma}{\pi}{\delta}{\xi}\Ra.
   \end{flalign*}}
 \end{definition}

In Figure~\ref{fig:event-partition}, we illustrate our notion of
matching using our running example of a discrete-time event
simulation system. Let the set of state variables be
$\{v_1,v_2\}$ and let the set of events contain $\{(e_1, 0),
(e_2,2)\}$, where event $e_i$ increments variable $v_i$ by 1. In the
figure, $\sigma$ is a fullpath of the concrete system and
$\delta$ is a fullpath of the abstract system.  (We only show a
prefix of the fullpaths.) The other parameter for
$\mathit{match}$ is $B$, which, for our example, is just the
identity relation. In order to show that
$\match{B}{\sigma}{\delta}$ holds, we have to find $\pi, \xi$
satisfying the definition. In the figure, we separate the
partitions induced by our choice for $\pi, \xi$ using
{\color{red}$--$} and connect elements related by $B$ with
{\color{orange}\line(1,0){10}}. Since all elements of a $\sigma$
partition are related to the first element of the corresponding
$\delta$ partition, $\corr{B}{\sigma}{\pi}{\delta}{\xi}$ holds,
therefore, $\match{B}{\sigma}{\delta}$ holds.

\begin{minipage}[t]{\textwidth}
\vspace{-5.5cm}
\hspace{-1.8cm} %
\includegraphics[scale=.6]{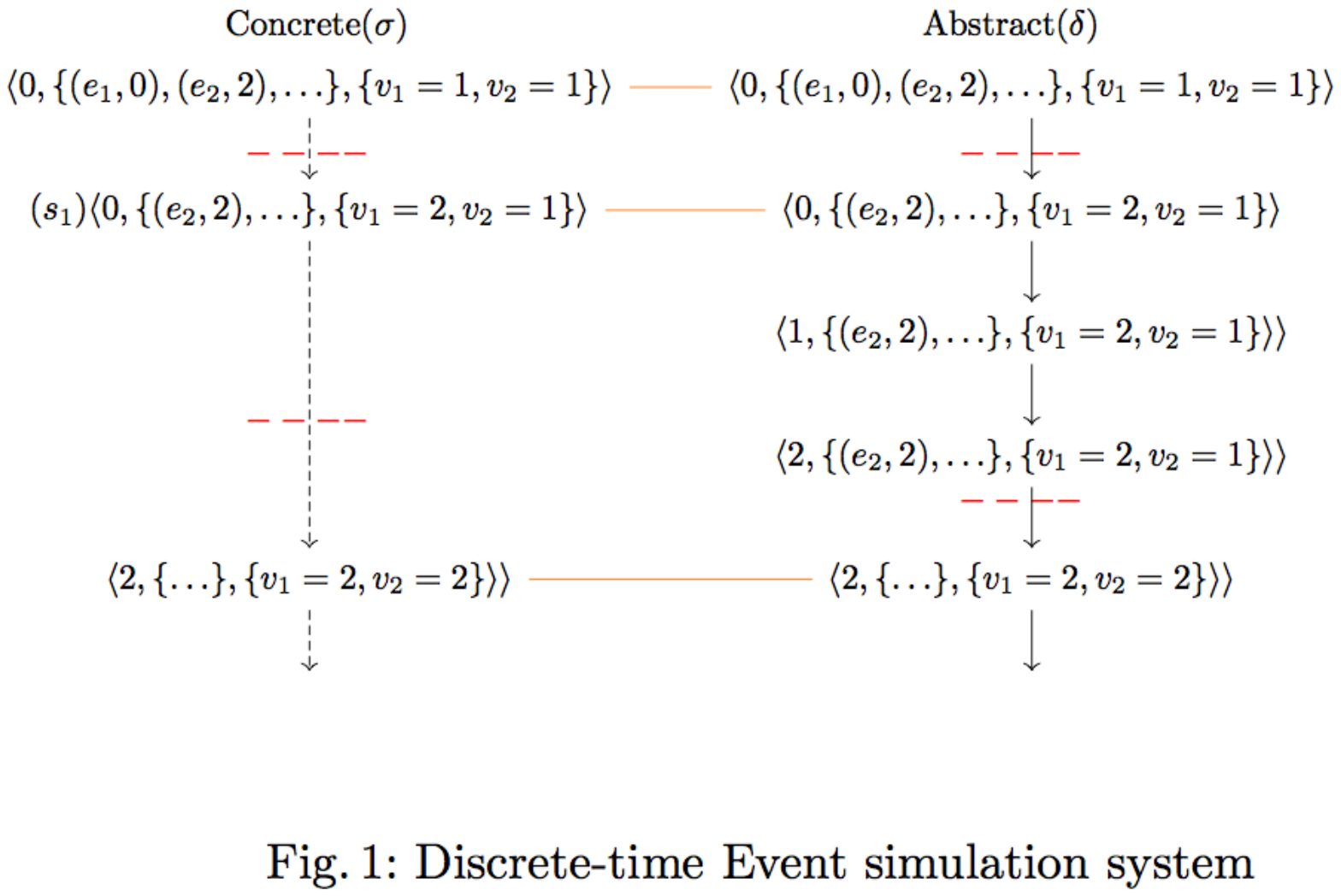}
\label{fig:event-partition}
\vspace{-4.5cm}
\end{minipage}

Given a labeled transition system \trs{}, a relation $B \subseteq S
\times S$ is a skipping simulation, if for any $s,w \in S$ such that
$sBw$, $s$ and $w$ are identically labeled and any fullpath starting
at $s$ can be matched by some fullpath starting at $w$.

\begin{definition}[Skipping Simulation]
  \label{def:sks}
  $B \subseteq S \times S$ is a skipping simulation (SKS) on TS \trs{}
  \myiff{} for all $s,w$ such that $sBw$,  the following hold.
        {\setlength{\abovedisplayskip}{3pt}
          \setlength{\belowdisplayskip}{5pt}
          \begin{flalign*}
            &\text{(SKS1) } L.s = L.w&\\ 
            & \text{(SKS2) } \langle \forall \sigma\from \fp.\sigma.s \from
            \langle \exists \delta \from \fp.\delta.w \from
            \match{B}{\sigma}{\delta} \rangle\rangle&
          \end{flalign*}
        }
\end{definition}

It may seem counter-intuitive to define skipping refinement with
respect to a single transition system, since our ultimate goal is
to relate transition systems at different levels of
abstraction. Our current approach has certain technical
advantages and we will see how to deal with two transitions
systems shortly.

In our running example of a discrete-time event simulation
system, neither the optimized concrete system nor the abstract
system stutter, \ie they do not require multiple steps to
complete the execution of an event. However, suppose that the
abstract and concrete system are modified so that execution of an
event takes multiple steps. For example, suppose that the
execution of $e_1$ in the concrete system (the first partition of
$\sigma$ in Figure~\ref{fig:event-partition}) takes 5 steps and
the execution of $e_1$ in the abstract system (the first
partition of $\delta$ in Figure~\ref{fig:event-partition}) takes
3 steps. Now, our abstract system is capable of stuttering and
the concrete system is capable of both stuttering and
skipping. Skipping simulation allows this, \ie we can define
$\pi,\xi$ such that $\corr{B}{\sigma}{\pi}{\delta}{\xi}$ still
holds.

Note that skipping simulation differs from weak
simulation~\cite{van1990linear}; the latter allows infinite
stuttering. Since we want to distinguish deadlock from
stuttering, it is important we distinguish between finite and
infinite stuttering. Skipping simulation also differs from
stuttering simulation, as skipping allows an implementation to
skip steps of the specification and therefore run ``faster'' than
the specification. In fact, skipping simulation is strictly
weaker than stuttering simulation.

\subsection{Skipping Refinement}
\label{sec:refinement}
We now show how the notion of skipping simulation, which is defined
in terms of a \emph{single} transition system, can be used to define
the notion of skipping refinement, a notion that relates \emph{two}
transition systems: an \emph{abstract} transition system and a
\emph{concrete} transition system. In order to define skipping
refinement, we make use of \emph{refinement maps}, functions that map
states of the concrete system to states of the abstract
system. Refinement maps are used to define what is observable at
concrete states. If the concrete system is a skipping refinement of
the abstract system, then its observable behaviors are also
behaviors of the abstract system, modulo skipping (which includes
stuttering). For example, in our running example, if the refinement
map is the identity function then any behavior of the optimized
system is a behavior of the abstract system modulo skipping.

\begin{definition}[Skipping Refinement]
  \label{def:skipref} \\
  Let \trs{A} and \trs{C} be transition systems and let $\mathit {r
  \from S_C \ra S_A}$ be a \emph{refinement map}.
  We say $\M{C}$ is a \textit{skipping refinement} of $\M{A}$ 
  with respect to $r$, written
  $\M{C} \lesssim_r \M{A}$, if there exists a relation $B \subseteq
  S_C \times S_A$ such that all of the following hold.
  \begin{enumerate}
  \item $\langle \forall s \in S_C :: sBr.s\rangle $ \emph{and}
  \item B is an SKS on \disjtrs{C}{A}{} where
    $\labf.s = L_A(s)$ for $s \in S_A$, and $\labf.s = L_A(r.s)$ for
    $s \in S_C$.
  \end{enumerate}
\end{definition}

Notice that we place no restrictions on refinement maps. When
refinement is used in specific contexts it is often useful to
place restrictions on what a refinement map can do, \eg we may
require for every $s \in S_C$ that $L_A(r.s)$ is a projection of
$L_C(s)$. Also, the choice of refinement map can have a big
impact on verification times~\cite{manolios2005computationally}.
Our purpose is to define a general theory of skipping, hence, we
prefer to be as permissive as possible.

\section{Automated Reasoning}
\label{sec:automated-reasoning}
To prove that transition system $\M{C}$ is a skipping refinement of
transition system $\M{A}$, we use Definitions~\ref{def:skipref}
and~\ref{def:sks}, which require us to show that for any fullpath
from $\M{C}$ we can find a ``matching'' fullpath from $\M{A}$.
However, reasoning about the existence of infinite sequences can be
problematic using automated tools. In order to avoid such reasoning,
we introduce the notion of well-founded skipping simulation. This
notion allows us to reason about skipping refinement by checking
mostly local properties, \ie properties involving states and their
successors. The intuition is, for any pair of states $s, w$, which
are related and a state $u$ such that $s \xrightarrow{} u$, there are
four cases to consider (Figure \ref{fig:wfsk}): (a) either we can match
the move from $s$ to $u$ right away, \ie there is a $v$ such that $w
\xrightarrow{} v$ and $u$ is related to $v$, or (b) there is
stuttering on the left, or (c) there is stuttering on the right, or
(d) there is skipping on the left.

\vspace{-8cm}
\begin{minipage}[c]{\textwidth}
\hspace{-2.5cm} %
\includegraphics[scale=.7]{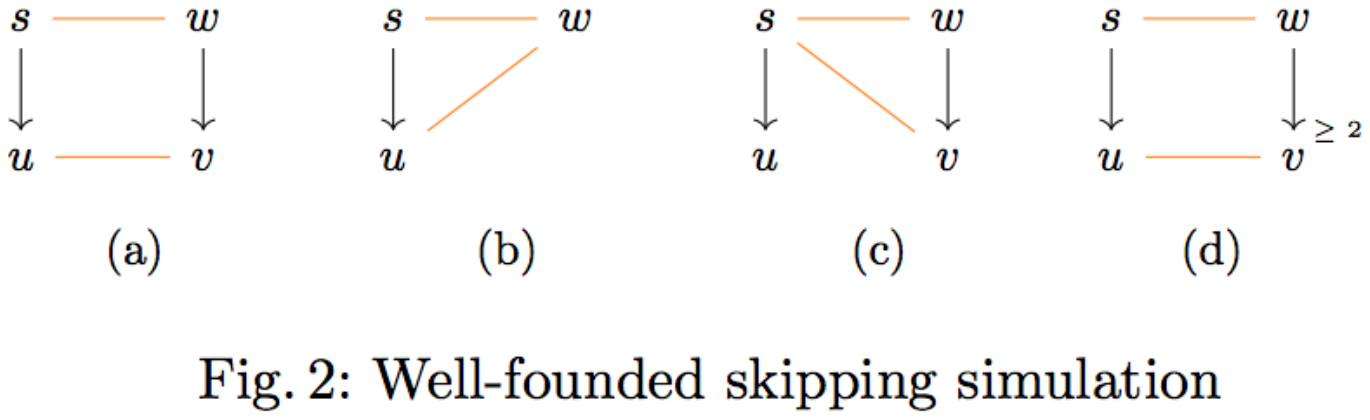}
\label{fig:wfsk}
\end{minipage}
\vspace{-8.5cm}

\begin{definition} [Well-founded Skipping]
  \label{def:wfsk}
  $B \subseteq S \times S $ is a well-founded skipping relation on TS
  \trs{} iff :
      {%
	\begin{enumerate}
	\item [(WFSK1)] $\langle \forall s,w \in S \from s B w \from L.s =
	  L.w\rangle$
	\item [(WFSK2)] There exist functions, $\rankt\from S\times S
	  \rightarrow W$, $\rankl \from S \times S \times S
	  \rightarrow \omega$, such that $\langle W, \prec \rangle$ is
	  well-founded and 
	  {\setlength{\abovedisplayskip}{3pt}
	    \setlength{\belowdisplayskip}{5pt}
	    \begin{flalign*}
              \La \forall s, & u,w \in S: s \xrightarrow{} u
              \wedge sBw: & \\
	      &\text{(a) } \langle \exists v \from w \xrightarrow{} v\from uBv \rangle \ \vee&\\
	      &\text{(b) } (uBw \wedge \rankt(u,w) \prec
	      \rankt(s,w)) \ \vee&\\
	      &\text{(c) } \langle \exists v \from w \xrightarrow{} v
	      \from s B v \wedge \rankl(v,s,u) < \rankl(w,s,u) \Ra \ \vee&\\
	      &\text{(d) } \langle \exists v : w \rightarrow^{\geq 2} v
	      \from uBv \Ra \Ra
	    \end{flalign*}
	  }
	\end{enumerate}
      }
\end{definition}

In the above definition, notice that condition (2d) requires us to
check that there exists a $v$ such that $v$ is \emph{reachable} from
$w$ and $uBv$ holds. Reasoning about reachability is not local in
general. However, for the kinds of optimized systems we are
interested in, we \emph{can} reason about reachability using local
methods because the number of abstract steps that a concrete step
corresponds to is bounded by a constant. As an example, the maximum
number of high-level steps that a concrete step of an optimized
memory controller can correspond to is the size of the request
buffer; this is a constant that is determined early in the
design. Another option is to replace condition (2d) with a condition
that requires only local reasoning. While this is possible, in light
of the above comments, the increased complexity is not justified.

Next, we show that the notion of well-founded skipping simulation
is equivalent to SKS and can be used as a sound and complete
proof rule to check if a given relation is an SKS. This allows us
to match infinite sequences by checking local properties and
bounded reachability. To show this we first introduce an
alternative definition for well-founded skipping simulation. The
motivation for doing this is that the alternate definition is
useful for proving the soundness and completeness theorems. It
also allows us to highlight the idea behind the conditions in
the definition of well-founded skipping simulation. The
simplification is based on two observations. First, it turns out
that (d) and (a) together subsume (c), so in the definition
below, we do not include case (c). Second, if instead of
$\rightarrow^{\geq 2}$ we use $\rightarrow^+$ in (d), then we
subsume case (a) as well.

\begin{definition}
  \label{def:rwfsk}
  $B \subseteq S \times S $ is a reduced well-founded skipping relation on
  TS \trs{} iff \from
  {%
    \begin{enumerate}
    \item [(RWFSK1)] $\langle \forall s,w \in S \from s B w \from L.s =
      L.w\rangle$
    \item [(RWFSK2)] There exists a function, $\rankt\from S\times S
      \rightarrow W$, such that $\langle W, \prec \rangle$ is well-founded and 
                  {\setlength{\abovedisplayskip}{3pt}
                    \setlength{\belowdisplayskip}{5pt}
                    \begin{flalign*}
                      \La \forall s, & u,w \in S: s \xrightarrow{} u \wedge
                      sBw: &\\
                      &\text{(a) } (uBw \wedge \rankt(u,w) \prec
                      \rankt(s,w)) \ \vee &\\
                      &\text{(b) } \langle \exists v : w \transplus v
                      \from uBv \Ra \Ra &
                    \end{flalign*}
                  }
    \end{enumerate}
  }
\end{definition}

In the sequel, ``WFSK'' is an abbreviation for ``well-founded
skipping relation'' and, similarly, ``RWFSK'' is an abbreviation for
``reduced well-founded skipping relation.''

We now show that WFSK and RWFSK are equivalent.

\begin{thm}
  \label{thm:wfsk-is-rwfsk}
  $B$ is a WFSK on \trs{} iff $B$ is an RWFSK on $\mathcal{M}$.
\end{thm}
\begin{proof}
  ($\Leftarrow$ direction): This direction is easy.

  \noindent ($\Rightarrow$ direction): 

  The key insight is that WFSK2c
  is redundant.

  Let $s,u, w \in S$, $s \trans u$, and $sBw$. If
  WFSK2a or WFSK2d holds then RWFSK2b holds. If WFSK2b holds, then
  RWFSK2a holds.  So, what remains is to assume that WFSK2c holds
  and neither of WFSK2a, WFSK2b, or WFSK2d hold. From this we will
  derive a contradiction.

  Let $\delta$ be a path starting at $w$, such that only WFSK2c
  holds between $s, u, \delta.i$. There are non-empty paths that
  satisfy this condition, \eg let $\delta = \La w \Ra$. In
  addition, any such path must be finite. If not, then for any
  adjacent pair of states in $\delta$, say $\delta.k$ and $\delta(k
  + 1)$, $\rankl(\delta(k+1),s,u) < \rankl(\delta.k,s,u)$, which
  contradicts the well-foundedness of $\rankl$. We also have that
  for every $k>0$, $u \centernot B \delta.k$; otherwise WFSK2a or
  WFSK2d holds. Now, let $\delta$ be a maximal path satisfying the
  above condition, \ie every extension of $\delta$ violates the
  condition. Let $x$ be the last state in $\delta$. We know that
  $sBx$ and only WFSK2c holds between $s,u,x$, so let $y$ be a
  witness for WFSK2c, which means that $sBy$ and one of WFSK2a,b,
  or d holds between $s,u,y$.  WSFK2b can't hold because then we
  would have $uBy$ (which would mean WFSK2a holds between
  $s,u,x$). So, one of WFSK2a,d has to hold, but that gives us a
  path from $x$ to some state $v$ such that $uBv$. The
  contradiction is that $v$ is also reachable from $w$, so WFSK2a
  or WFSK2d held between $s,u,w$. \ffp
\end{proof}

Let's now discuss why we included condition WFSK2c. The systems
we are interested in verifying have a bound---determined early
early in the design---on the number of skipping steps
possible. The problem is that RWSFK2b forces us to deal with
stuttering and skipping steps in the same way, while with WFSK
any amount of stuttering is dealt with locally. Hence, WFSK
should be used for automated proofs and RWFSK can be used for
meta reasoning.

One more observation is that the proof of
Theorem~\ref{thm:wfsk-is-rwfsk}, by showing that WFSK2c is
redundant, highlights why skipping refinement subsumes stuttering
refinement. Therefore, skipping refinement is a weaker, but more
generally applicable notion of refinement than stuttering
refinement.

In what follows, we show that the notion of RWFSK (and by
Theorem~\ref{thm:wfsk-is-rwfsk} WFSK) is equivalent to SKS and
can be used as a sound and complete proof rule to check if a
given relation is an SKS. This allows us to match infinite
sequences by checking local properties and bounded
reachability. We first prove soundness, \ie any RWFSK is an
SKS. The proof proceeds by showing that given a RWFSK relation
$B$, $sBw$, and any fullpath starting at $s$, we can recursively
construct a fullpath $\delta$ starting at $w$, and increasing
sequences $\pi,\xi$ such that fullpath at $s$ matches $\delta$.

\begin{thm}[Soundness]
  \label{thm:WFSKsoundness}
  If $B$ is an RWFSK on \M{} then $B$ is a SKS on \M{}.
\end{thm}

\begin{proof}
  To show that $B$ is an SKS on \trs{}, we show that given $B$ is a
  RWFSK on \trs{} and $x, y \in S$ such that $xBy$, SKS1 and SKS2 hold.
  SKS1 follows directly from condition 1 of RWSFK.

  Next we show that SKS2 holds. We start by recursively defining
  $\delta$. In the process, we also define partitions $\pi$ and
  $\xi$. For the base case, we let $\pi.0 = 0$, $\xi.0 = 0$ and
  $\delta.0 = y$. By assumption $\sigma(\pi.0) B \delta(\xi.0)$. For
  the recursive case, assume that we have defined $\pi.0, \ldots,
  \pi.i$ as well as $\xi.0, \ldots, \xi.i$ and $\delta.0, \ldots,
  \delta(\xi.i)$. We also assume that $\sigma(\pi.i) B \delta(\xi.i)$.
  Let $s$ be $\sigma(\pi.i)$; let $u$ be $\sigma(\pi.i + 1)$; let $w$
  be $\delta(\xi.i)$.  We consider two cases.

  First, say that RWFSK2b holds. Then, there is a $v$ such that $w
  \transplus v$ and $uBv$.  Let $\vv{v} = [v_0 = w, \ldots, v_m =
    v]$ be a finite path from $w$ to $v$ where $m \geq 1$. We
  define $\pi(i+1) = \pi.i + 1, \xi(i+1) = \xi.i + m$, $
  ^\xi\delta^i = [v_0, \ldots, v_{m-1}]$ and $\delta(\xi(i + 1)) = v$.

  If the first case does not hold, \ie RWFSK2b does not hold, and
  RWFSK2a does hold.  We define $J$ to be the subset of the positive
  integers such that for every $j \in J$, the following holds.
  \begin{align}
    \label{cond:sound}
    & \La \forall v : w \transplus v : \neg (\sigma(\pi.i + j) B v) \Ra
    \ \wedge \\
    \sigma(\pi.i + j) B w  \ & \wedge \ \rankt(\sigma(\pi.i + j), w) \prec
    \rankt(\sigma(\pi.i + j - 1) ,w) \nonumber
  \end{align}

  The first thing to observe is that $1 \in J$ because
  $\sigma(\pi.i + 1) = u$, RWFSK2b does not hold (so the first
  conjunct is true) and RWFSK2a does (so the second conjunct is
  true).  The next thing to observe is that there exists a
  positive integer $n>1$ such that $n \not\in J$. Suppose not,
  then for all $n \geq 1, n \in J$. Now, consider the (infinite)
  suffix of $\sigma$ starting at $\pi.i$.  For every adjacent
  pair of states in this suffix, say $\sigma(\pi.i + k)$ and
  $\sigma(\pi.i+k+1)$ where $k \geq 0$, we have that
  $\sigma(\pi.i+k)Bw$ and that only RWFSK2a applies (\ie RWFSK2b
  does not apply). This gives us a contradiction because $\rankt$
  is well-founded.  We can now define $n$ to be $\mathit{min}(\{l
  : l \not\in J\})$. Notice that only RWFSK2a holds between
  $\sigma(\pi.i + n-1)), \sigma(\pi.i + n)$ and $w$, hence
  $\sigma(\pi.i + n)Bw$ and $\rankt(\sigma(\pi.i+n), w) \prec
  \rankt(\sigma(\pi.i+n-1),w)$. Since Formula~\ref{cond:sound}
  does not hold for $n$, there is a $v$ such that $w \transplus v
  \wedge \sigma(\pi.i + n) B v$. Let $\vv{v} = [v_0 = w, \ldots,
    v_m = v]$ be a finite path from $w$ to $v$ where $m \geq
  1$. We are now ready to extend our recursive definition as
  follows: $\pi(i+1) = \pi.i+n$, $\xi(i+1) = \xi.i + m$, and
  $^\xi\delta^i = [v_0, \ldots, v_{m-1}]$.

Now that we defined $\delta$ we can show that SKS2 holds.  We
start by unwinding definitions. The first step is to show that
$\fp.\delta.y$ holds, which is true by construction.  Next, we
show that $\match{B}{\sigma}{\delta}$ by unwinding the definition
of $\mathit{match}$. That involves showing that there exist $\pi$
and $\xi$ such that $\corr{B}{\sigma}{\pi}{\delta}{\xi}$
holds. The $\pi$ and $\xi$ we used to define $\delta$ can be used
here. Finally, we unwind the definition of $\mathit{corr}$, which
gives us a universally quantified formula over the natural
numbers. This is handled by induction on the segment index; the
proof is based on the recursive definitions given above. \ffp

\end{proof}

We next state completeness, \ie given a SKS relation $B$ we provide as
witness a well-founded structure $\La W,\prec \Ra$, and a rank
function $\rankt$ such that the conditions in
Definition~\ref{def:rwfsk} hold.

\begin{thm}[Completeness]
  \label{thm:WFSKCompleteness} 
  If B is an SKS on \M{}, then B is an RWFSK on \M{}.
\end{thm}

The proof requires us to introduce a few definitions and lemmas.

\begin{definition}
  Given TS \trs{}, the \emph{computation tree} rooted at a state
  $s \in S$, denoted $\ctree(\M{},s)$, is obtained by
  ``unfolding'' \M{} from $s$. Nodes of $\ctree(\M{},s)$ are finite
  sequences over $S$ and $\ctree(\M{},s)$ is the smallest tree
  satisfying the following.
  \begin{enumerate}
  \item  The root is $\langle s \rangle$.
  \item If $\langle s, \ldots, w \rangle$ is a node and $w
    \xrightarrow{} v$, then $\langle s, \ldots, w, v \rangle$ is a
    node whose parent is $\langle s, \ldots, w \rangle$.
  \end{enumerate}
\end{definition}

Our next definition is used to construct the ranking function
appearing in the definition of RWFSK.

\begin{definition} (ranktCt)
  Given an SKS $B$, if $\neg(s B w)$, then $\ranktct(\M{},s,w)$ is the
  empty tree, otherwise $\ranktct(\M{},s,w)$ is the largest subtree of
  $\ctree(\M{},s)$ such that for any non-root node of
  $\ranktct(\M{},s,w)$, $\langle s, \ldots, x \rangle$, we have that
  $x B w$ and $\langle \forall v : w \transplus v: \neg(x B v)
  \rangle$.
\end{definition}

A basic property of our construction is the finiteness of paths.

\begin{lem}
  \label{lemma:finite-branch-tree}
  Every path of $\ranktct(\M{},s,w)$ is finite.
\end{lem}

Given Lemma~\ref{lemma:finite-branch-tree}, we define a function,
$\size$, that given a tree, $t$, all of whose paths are finite,
assigns an ordinal to $t$ and to all nodes in $t$. The ordinal
assigned to node $x$ in $t$ is defined as follows: $\size(t,x) =
\bigcup_{c \in \mathit{children}.x} \size(t,c) +1$. We are using
set theory, \eg an ordinal number is defined to be the set of
ordinal numbers below it, which explains why it makes sense to
take the union of ordinal numbers. The size of a tree is the size
of its root, \ie $\size(\ranktct(\M{},s,w)) =
\size(\ranktct(\M{},s,w),\La s \Ra)$. We use $\preceq$ to compare
ordinal and cardinal numbers.

\begin{lem}
  \label{lemma:countable-label-tree} 
  If $|S| \preceq \kappa$, where $\omega \preceq \kappa$ then for all
  $s,w \in S$, $\size(\ranktct(\M{},s,w))$ is an ordinal of cardinality
  $\preceq \kappa$.
\end{lem}

Lemma~\ref{lemma:countable-label-tree} shows that we can use as the
domain of our well-founded function in RWFSK2 the cardinal
$\mathit{max}(|S|^+, \omega)$: either $\omega$ if the state space is
finite, or $|S|^+$, the cardinal successor of the size of the state
space otherwise.

\begin{lem} 
  \label{lem:tree}
  If $sBw, s \xrightarrow{} u, u \in \ranktct(\M{},s,w)$ then
  $\size(\ranktct(\M{},u,w)) \prec \size(\ranktct(\M{},s,w))$.
\end{lem}

\noindent We are now ready to prove completeness.
\begin{proof}(Completeness)
  We assume that $B$ is an SKS on $\M{}$ and we show that this
  implies that $B$ is also an RWFSK on $\M{}$. RWFSK1 follows
  directly.  To show that RWFSK2 holds, let $W$ be the successor
  cardinal of $\mathit{max}(|S|, \omega)$ and let $\rankt(a,b)$ be
  $\size(\ranktct(\M{},a,b))$. Given $s, u, w \in S$ such that $s
  \trans u$ and $sBw$, we show that either RWFSK2(a) or RWFSK2(b)
  holds.

  There are two cases.  First, suppose that $\La \exists v : w
  \transplus v : u B v\Ra$ holds, then RWFSK2(b) holds. If not,
  then $\La \forall v : w \transplus v : \neg(u B v) \Ra$, but $B$
  is an SKS so let $\sigma$ be a fullpath starting at $s, u$. Then
  there is a fullpath $\delta$ such that $\fp.\delta.w$ and
  $\match{B}{\sigma}{\delta}$. Hence, there exists $\pi, \xi \in
  \inc$ such that $\corr{B}{\sigma}{\pi}{\delta}{\xi}$. By the
  definition of $\mathit{corr}$, we have that $uB\delta(\xi.i)$ for
  some $i$, but $i$ cannot be greater than $0$ because then $uBx$
  for some $x$ reachable from $w$, violating the assumptions of
  the case we are considering. So, $i=0$, \ie $uBw$. By lemma~\ref{lem:tree},
  $\rankt(u,w) = \size(\ranktct(\M{},u,w)) \prec
  \size(\ranktct(\M{},s,w)) = \rankt(s,w). \ffp$
 \end{proof}

Following Abadi and Lamport~\cite{lamport1991existence}, one of
the basic questions asked about new notions of refinement is:
under what conditions do refinement maps exist?  Abadi and
Lamport required several rather complex conditions, but our
completeness proof shows that for skipping refinement, refinement
maps always exist. See Section~\ref{sec:related} for more
information.

Well-founded skipping gives us a simple proof rule to determine if a
concrete transition system \M{C} is a skipping refinement of an
abstract transition system \M{A} with respect to a refinement map
$r$. Given a refinement map $r : S_C \rightarrow S_A$ and relation
$B \subseteq S_C \times S_A$, we check the following two conditions: (a)
for all $s \in S_C$, $sBr.s$ and (b) if $B$ is a WFSK on disjoint
union of \M{C} and \M{A}. If (a) and (b) hold, from
~\autoref{thm:WFSKsoundness}, $\M{C} \lesssim_r \M{A}$.

\section{Experimental Evaluation}
\label{sec:casestudies}
In this section, we experimentally evaluate the theory of
skipping refinement using three case studies: a JVM-inspired
stack machine, an optimized memory controller, and a
vectorization compiler transformation.  Our goals are to evaluate
the specification costs and benefits of using skipping refinement
as a notion of correctness and to determine the impact that the
use of skipping refinement has on state-of-the-art verification
tools in terms of capacity and verification times.  We do that by
comparing the cost of proving correctness using skipping
refinement with the cost of using input-output equivalence: if
the specification and the implementation systems start in
equivalent initial states and get the same inputs, then if both
systems terminate, the final states of the systems are also
equivalent. We chose I/O equivalence since that is the most
straightforward way of using existing tools to reason about our
case studies. Since skipping simulation is a stronger notion of
correctness that I/O equivalence, skipping proofs provide more
information, \eg I/O equivalence holds even if the concrete
system diverges, but skipping simulation does not hold and would
therefore catch such divergence errors.

The first two case studies were developed and compiled to
sequential AIGs using the BAT tool~\cite{manolios2007bat}, and
then analyzed using the TIP, IIMC, BLIMC, and SUPER\_PROVE
model-checkers~\cite{hwmcc13-results}.  SUPER\_PROVE and
IIMC  are the top performing model-checkers in the single
safety property track of the Hardware Model Checking
Competition~\cite{hwmcc13-results}. We chose TIP and BLIMC to
cover tools based on temporal decomposition and bounded
model-checking.  The last case study involves systems whose state
space is infinite.  Since model checkers cannot be used to verify
such systems, we used the ACL2s interactive theorem
prover~\cite{acl2s11}. BAT files, corresponding AIGs, ACL2s
models, and ACL2s proof scripts are publicly
available~\cite{sksmodel}, hence we only briefly describe the
case studies.

Our results show that with I/O equivalence, model-checkers
quickly start timing out as the complexity of the systems
increases. In contrast, with skipping refinement much larger
systems can be automatically verified. For the infinite state
case study, interactive theorem proving was used and the manual
effort required to prove skipping refinement theorems was
significantly less than the effort required to prove I/O
equivalence.

\noindent
\emph{JVM-inspired Stack Machine.}
For this case study we defined BSTK, a simple hardware
implementation of part of Java Virtual Machine
(JVM)~\cite{hardin2001real}. BSTK models an instruction memory,
an instruction buffer and a stack. It supports a small subset of
JVM instructions, including \mmit{push, pop, top, nop}.
STK is the high-level specification with respect to which we verify
the correctness of BSTK. The state of STK consists
of an instruction memory (\mmit{imem}), a program counter (\mmit{pc}),
and a stack (\mmit{stk}). STK fetches an instruction from the
\mmit{imem}, executes it, increases the \mmit{pc} and possibly
modifies the \mmit{stk}.
The state of BSTK is similar to STK, except that it also includes
an instruction buffer, whose capacity is a parameter. BSTK
fetches an instruction from the \mmit{imem} and as long as the
fetched instruction is not \mmit{top} and the instruction buffer
(\mmit{ibuf}) is not full, it enqueues it to the end of the
\mmit{ibuf} and increments the \mmit{pc}. If the fetched
instruction is \mmit{top} or \mmit{ibuf} is full, the machine
executes all buffered instructions in the order they were
enqueued, thereby draining the \mmit{ibuf} and obtaining a new
\mmit{stk}.

\noindent
\emph{Memory Controller.}
We defined a memory controller, OptMEMC, which fetches a memory
request from location \mmit{pt} in a queue of CPU requests,
\mmit{reqs}. It enqueues the fetched request in the request
buffer, $\mathit{rbuf}$ and increments \mmit{pt} to point to the
next CPU request in \mmit{reqs}. If the fetched request is a
\mmit{read} or the request buffer is full (the capacity of
\mmit{rbuf} is parameter), then before enqueuing the request into
\mmit{rbuf}, OptMEMC first analyzes the request buffer for
consecutive write requests to the same address in the memory
(\mmit{mem}). If such a pair of writes exists in the buffer, it
marks the older write requests in the request buffer as
redundant. Then it executes all the requests in the request
buffer except the marked (redundant) ones. Requests in the buffer
are executed in the order they were enqueued. We also defined
MEMC, a specification system that processes each memory request
atomically.

\noindent
\emph{Results.}
To evaluate the computational benefits of skipping refinement, we
created a benchmark suite including versions of the BSTK and STK
machines---parameterized by the size of \mmit{imem}, \mmit{ibuf}, and
\mmit{stk}---and OptMEMC and MEMC machines---parameterized by the size
of \mmit{req, rbuf} and \mmit{mem}. These models had anywhere from 24K
gates and 500 latches to 2M gates and 23K latches. We used a machine
with an Intel Xeon X5677 with 16 cores running at 3.4GHz and 96GB main
memory.  The timeout limit for model-checker runs is set to 900
seconds.  In Figure~\ref{fig:runningtime}, we plot the running times
for the four model-checkers used. The $x$-axis represents the running
time using I/O equivalence and $y$-axis represents the running time
using skipping refinement. A point with $x = $ TO indicates that the
model-checker timed out for I/O equivalence while $y = $ TO indicates
that the model-checker timed out for skipping refinement. Our results
show that model-checkers timeout for most of the configurations when
using I/O equivalence while all model-checkers except TIP can solve
all the configurations using skipping refinement. Furthermore, there
is an improvement of several orders of magnitude in the running time
when using skipping refinement. The performance benefits are partly
due to the structure provided by the skipping refinement proof
obligation. For example, we have a bound on the number of steps that
the optimized systems can skip before a match occurs and we have rank
functions for stuttering. This allows the model checkers to locally
check correctness instead of having to prove correspondence at the
input/output boundaries, as is the case for I/O equivalence.

\begin{figure}[tb]
\begin{minipage}[]{\textwidth}
  \centering
  \scalebox{.6}{
    \input{combined-stack-mem.tikz}
  }
  \captionof{figure}{Performance of model-checkers on case studies}
  \label{fig:runningtime}
  \vspace{.1cm}
\end{minipage}
\end{figure}
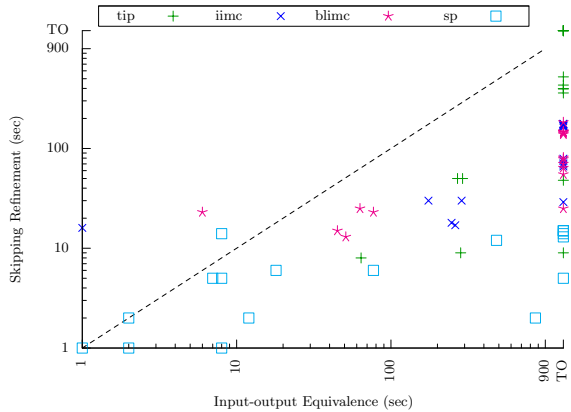

\paragraph{Superword-level Parallelism with SIMD instructions.}
For this case study we verify the correctness of a compiler
transformation from a source language containing only scalar
instructions to a target language containing both scalar and
vector instructions. We model the transformation as a function
that given a program in the source language and generates a
program in the target language. We use the translation validation
approach to compiler correctness and prove that the target
program implements the source program~\cite{barrett2005tvoc}.

For presentation purposes, we make some simplifying assumptions: the
state of the source and target programs (modeled as transition
systems) is a tuple consisting of a sequence of instructions, a
program counter and a store. We also assume that a SIMD instruction
operates on two sets of data operands simultaneously and that
the transformation identifies parallelism at the basic block
level. Therefore, we do not consider control flow.

For this case study, we used deductive verification methodology to
prove correctness. The scalar and vector machines are defined
using the data-definition framework in
ACL2s~\cite{acl2s11,acl2sweb,defdata}.  We formalized the
operational semantics of the scalar and vector machines using
standard methods. The sizes of the program and store are
unbounded and thus the state space of the machines is
infinite. Once the definitions were in place, proving skipping
refinement with ACL2s was straightforward. Proving I/O
equivalence requires significantly more theorem proving expertise
and insight to come up with the right invariants, something we
avoided with the skipping proof. The proof scripts are publicly
available~\cite{sksmodel}.

\section{Related Work and Discussion}
\label{sec:related}

\emph{Notions of correctness.} Notions of correctness for
reasoning about reactive systems have been widely studied and we
refer the reader to excellent surveys on this
topic~\cite{pnueli1985linear,van1990linear,lynch1996forward}.
Lamport~\cite{lamport1983good} argues that abstract and the
concrete systems often only differ by stuttering steps; hence a
notion of correctness should directly account stuttering. Weak
simulation~\cite{van1990linear} and stuttering
simulation~\cite{manolios2003compositional} are examples of such
notions. These notions are too strong to reason about optimized
reactive systems, hence the need for skipping refinement, which
allows both stuttering \emph{and} skipping.

\emph{Refinement Maps.} A basic question in a theory of
refinement is whether refinement maps exist: if  a concrete
system implements an abstract system, does there
exists a refinement map that can be use to prove it? Abadi and
Lamport~\cite{lamport1991existence} showed that in the
linear-time framework, a refinement map exists provided the systems
satisfy a number of complex
conditions. In~\cite{manolios2001mechanical}, it was shown that
for STS, a branching-time notion, the existence of refinement maps
does not depend on any of the conditions found in the work of
Abadi and Lamport and that this result can be extended to the
linear-time case~\cite{manolios2003compositional}. We also show
that for skipping refinement, refinement maps always exists.

\emph{Hardware Verification.} Several approaches to verification of
superscalar processors appear in the literature and as new
features are modeled new variants of correctness notions are
proposed~\cite{aagaard2001framework}. These variants can be
broadly classified on the basis of whether (1) they support
nondeterministic abstracts systems or not (2) they support
nondeterministic concrete systems or not (3) the kinds of
refinement maps allowed. The theory of skipping refinement
provides a general framework that support nondeterministic
abstract and concrete systems and arbitrary refinement maps. We
believe that a uniform notion of correctness can significantly
ease the verification effort.

\emph{Software Verification.}  Program refinement is widely used to
verify the correctness of programs and program transformations.
Several back-end compiler transformations are proven correct in
CompCert~\cite{leroy2009formally} by showing that the source and the
target language of a transformation are related by the notion of
\emph{forward simulation}.  In~\cite{namjoshi2013witnessing}, several
compiler transformations, \eg dead-code elimination and control-flow
graph compression, are analyzed using a more general notion of
refinement based on stuttering simulation.

Like CompCert, the semantics of the source and target languages
are assumed to be deterministic and the only source of
non-determinism comes from initial states. In section
\ref{sec:casestudies}, we used skipping refinement and a
methodology similar to translation
validation~\cite{barrett2005tvoc} to analyze a compiler
transformation that extracts superword parallelism in a
program. It is not possible to prove the correctness of this
transformation using stuttering refinement.
In~\cite{dockins2012operational}, \emph{choice refinement} is
introduced to account for compiler transformations that resolve
internal nondeterministic choices in the semantics of the source
language (\eg the left-to-right evaluation strategy). Skipping
refinement is an appropriate notion of correctness to analyze
such transformations.  In~\cite{manolios2006framework}, it is shown
how to prove the correctness of assembly programs running on a
pipelined machine by first proving that the assembly code is
correct when running on an idealized processor and, second, by
proving that the pipelined machine is a refinement of idealize
processor. Skipping can be similarly used to combine hardware and
software verification for optimized systems.

\section{Conclusion and Future Work}
\label{sec:conclusions}

In this paper, we introduced skipping refinement, a new notion of
correctness for reasoning about optimized reactive systems where the
concrete implementation can execute faster than its
specification. This is the first notion of refinement that we know of
that can directly deal with such optimized systems. We presented a
sound and complete characterization of skipping that is local, \ie for
the kinds of systems we consider, we can prove skipping refinement
theorems by reasoning only about paths whose length is bounded by a
constant. This characterization provides a convenient proof method and
also enables mechanization and automated verification. We
experimentally validated skipping refinement and our local
characterization by performing three case studies. Our experimental
results show that, for relatively simple configurations, proving
correctness directly, without using skipping, is beyond the
capabilities of current model-checking technology, but when using
skipping refinement, current model-checkers are able to prove
correctness.  For future work, we plan to characterize the class of
temporal properties preserved by skipping refinement, to develop and
exploit compositional reasoning for skipping refinement, and to use
skipping refinement for testing-based verification and validation.

\bibliographystyle{splncs03}
\bibliography{skipping-refinement}

\end{document}

%% file: combined-stack-mem.tikz
\begin{tikzpicture}[gnuplot]
\gpsolidlines
\gpcolor{gp lt color border}
\gpsetlinetype{gp lt border}
\gpsetlinewidth{1.00}
\draw[gp path] (1.688,0.985)--(1.868,0.985);
\node[gp node right] at (1.504,0.985) { \small{1}};
\draw[gp path] (1.688,1.651)--(1.778,1.651);
\draw[gp path] (1.688,2.041)--(1.778,2.041);
\draw[gp path] (1.688,2.317)--(1.778,2.317);
\draw[gp path] (1.688,2.532)--(1.778,2.532);
\draw[gp path] (1.688,2.707)--(1.778,2.707);
\draw[gp path] (1.688,2.855)--(1.778,2.855);
\draw[gp path] (1.688,2.983)--(1.778,2.983);
\draw[gp path] (1.688,3.096)--(1.778,3.096);
\draw[gp path] (1.688,3.198)--(1.868,3.198);
\node[gp node right] at (1.504,3.198) { \small{10}};
\draw[gp path] (1.688,3.864)--(1.778,3.864);
\draw[gp path] (1.688,4.253)--(1.778,4.253);
\draw[gp path] (1.688,4.530)--(1.778,4.530);
\draw[gp path] (1.688,4.744)--(1.778,4.744);
\draw[gp path] (1.688,4.919)--(1.778,4.919);
\draw[gp path] (1.688,5.068)--(1.778,5.068);
\draw[gp path] (1.688,5.196)--(1.778,5.196);
\draw[gp path] (1.688,5.309)--(1.778,5.309);
\draw[gp path] (1.688,5.410)--(1.868,5.410);
\node[gp node right] at (1.504,5.410) { \small{100}};
\draw[gp path] (1.688,6.076)--(1.778,6.076);
\draw[gp path] (1.688,6.466)--(1.778,6.466);
\draw[gp path] (1.688,6.742)--(1.778,6.742);
\draw[gp path] (1.688,6.957)--(1.778,6.957);
\draw[gp path] (1.688,7.132)--(1.778,7.132);
\draw[gp path] (1.688,7.280)--(1.778,7.280);
\draw[gp path] (1.688,7.409)--(1.778,7.409);
\draw[gp path] (1.688,7.522)--(1.778,7.522);
\draw[gp path] (1.688,7.623)--(1.868,7.623);
\node[gp node right] at (1.504,7.623) { \small{900}};
\draw[gp path] (1.688,8.023)--(1.868,8.023);
\node[gp node right] at (1.504,8.023) { \small{TO}};
\draw[gp path] (1.688,0.985)--(1.688,1.165);
\node[gp node right,rotate=-270] at (1.688,0.857) { \small{1}};
\draw[gp path] (2.717,0.985)--(2.717,1.075);
\draw[gp path] (3.320,0.985)--(3.320,1.075);
\draw[gp path] (3.747,0.985)--(3.747,1.075);
\draw[gp path] (4.078,0.985)--(4.078,1.075);
\draw[gp path] (4.349,0.985)--(4.349,1.075);
\draw[gp path] (4.578,0.985)--(4.578,1.075);
\draw[gp path] (4.776,0.985)--(4.776,1.075);
\draw[gp path] (4.951,0.985)--(4.951,1.075);
\draw[gp path] (5.108,0.985)--(5.108,1.165);
\node[gp node right,rotate=-270] at (5.108,0.857) { \small{10}};
\draw[gp path] (6.137,0.985)--(6.137,1.075);
\draw[gp path] (6.739,0.985)--(6.739,1.075);
\draw[gp path] (7.167,0.985)--(7.167,1.075);
\draw[gp path] (7.498,0.985)--(7.498,1.075);
\draw[gp path] (7.769,0.985)--(7.769,1.075);
\draw[gp path] (7.998,0.985)--(7.998,1.075);
\draw[gp path] (8.196,0.985)--(8.196,1.075);
\draw[gp path] (8.371,0.985)--(8.371,1.075);
\draw[gp path] (8.527,0.985)--(8.527,1.165);
\node[gp node right,rotate=-270] at (8.527,0.857) { \small{100}};
\draw[gp path] (9.557,0.985)--(9.557,1.075);
\draw[gp path] (10.159,0.985)--(10.159,1.075);
\draw[gp path] (10.586,0.985)--(10.586,1.075);
\draw[gp path] (10.918,0.985)--(10.918,1.075);
\draw[gp path] (11.188,0.985)--(11.188,1.075);
\draw[gp path] (11.417,0.985)--(11.417,1.075);
\draw[gp path] (11.616,0.985)--(11.616,1.075);
\draw[gp path] (11.791,0.985)--(11.791,1.075);
\draw[gp path] (11.947,0.985)--(11.947,1.165);
\node[gp node right,rotate=-270] at (11.947,0.857) { \small{900}};
\draw[gp path] (12.347,0.985)--(12.347,1.165);
\node[gp node right,rotate=-270] at (12.347,0.857) { \small{TO}};
\draw[gp path] (1.688,7.623)--(1.688,0.985)--(12.347,0.985);
\node[gp node center,rotate=-270] at (0.246,4.304) {Skipping Refinement (sec)};
\node[gp node center] at (6.817,-0.215) {Input-output Equivalence (sec)};
\draw[gp path] (2.049,8.120)--(2.049,8.570)--(11.585,8.570)--(11.585,8.120)--cycle;
\draw[gp path] (2.049,8.570)--(11.585,8.570);
\gpcolor{gp lt color axes}
\gpsetlinetype{gp lt axes}
\draw[gp path] (1.688,0.985)--(1.792,1.052)--(1.895,1.119)--(1.999,1.186)--(2.103,1.253)%
  --(2.206,1.320)--(2.310,1.387)--(2.413,1.454)--(2.517,1.521)--(2.621,1.588)--(2.724,1.656)%
  --(2.828,1.723)--(2.932,1.790)--(3.035,1.857)--(3.139,1.924)--(3.242,1.991)--(3.346,2.058)%
  --(3.450,2.125)--(3.553,2.192)--(3.657,2.259)--(3.761,2.326)--(3.864,2.393)--(3.968,2.460)%
  --(4.071,2.527)--(4.175,2.594)--(4.279,2.661)--(4.382,2.728)--(4.486,2.795)--(4.590,2.862)%
  --(4.693,2.929)--(4.797,2.997)--(4.900,3.064)--(5.004,3.131)--(5.108,3.198)--(5.211,3.265)%
  --(5.315,3.332)--(5.419,3.399)--(5.522,3.466)--(5.626,3.533)--(5.729,3.600)--(5.833,3.667)%
  --(5.937,3.734)--(6.040,3.801)--(6.144,3.868)--(6.248,3.935)--(6.351,4.002)--(6.455,4.069)%
  --(6.558,4.136)--(6.662,4.203)--(6.766,4.270)--(6.869,4.338)--(6.973,4.405)--(7.077,4.472)%
  --(7.180,4.539)--(7.284,4.606)--(7.387,4.673)--(7.491,4.740)--(7.595,4.807)--(7.698,4.874)%
  --(7.802,4.941)--(7.906,5.008)--(8.009,5.075)--(8.113,5.142)--(8.216,5.209)--(8.320,5.276)%
  --(8.424,5.343)--(8.527,5.410)--(8.631,5.477)--(8.735,5.544)--(8.838,5.611)--(8.942,5.679)%
  --(9.045,5.746)--(9.149,5.813)--(9.253,5.880)--(9.356,5.947)--(9.460,6.014)--(9.564,6.081)%
  --(9.667,6.148)--(9.771,6.215)--(9.874,6.282)--(9.978,6.349)--(10.082,6.416)--(10.185,6.483)%
  --(10.289,6.550)--(10.393,6.617)--(10.496,6.684)--(10.600,6.751)--(10.703,6.818)--(10.807,6.885)%
  --(10.911,6.952)--(11.014,7.020)--(11.118,7.087)--(11.222,7.154)--(11.325,7.221)--(11.429,7.288)%
  --(11.532,7.355)--(11.636,7.422)--(11.740,7.489)--(11.843,7.556)--(11.947,7.623);
\gpcolor{gp lt color border}
\node[gp node right] at (2.969,8.345) {tip};
\gpcolor{gp lt color 1}
\gpsetpointsize{8.00}
\gppoint{gp mark 1}{(7.865,2.983)}
\gppoint{gp mark 1}{(10.002,4.744)}
\gppoint{gp mark 1}{(12.347,6.641)}
\gppoint{gp mark 1}{(12.347,6.738)}
\gppoint{gp mark 1}{(12.347,6.816)}
\gppoint{gp mark 1}{(12.347,6.730)}
\gppoint{gp mark 1}{(12.347,6.998)}
\gppoint{gp mark 1}{(12.347,8.023)}
\gppoint{gp mark 1}{(12.347,8.023)}
\gppoint{gp mark 1}{(12.347,8.023)}
\gppoint{gp mark 1}{(12.347,3.096)}
\gppoint{gp mark 1}{(10.072,3.096)}
\gppoint{gp mark 1}{(12.347,4.705)}
\gppoint{gp mark 1}{(10.114,4.744)}
\gppoint{gp mark 1}{(12.347,8.023)}
\gppoint{gp mark 1}{(12.347,8.023)}
\gppoint{gp mark 1}{(12.347,8.023)}
\gppoint{gp mark 1}{(12.347,8.023)}
\gppoint{gp mark 1}{(3.701,8.345)}
\gpcolor{gp lt color border}
\node[gp node right] at (5.353,8.345) {iimc};
\gpcolor{gp lt color 2}
\gppoint{gp mark 2}{(1.688,3.649)}
\gppoint{gp mark 2}{(12.347,4.221)}
\gppoint{gp mark 2}{(12.347,5.011)}
\gppoint{gp mark 2}{(12.347,5.081)}
\gppoint{gp mark 2}{(12.347,5.134)}
\gppoint{gp mark 2}{(12.347,5.054)}
\gppoint{gp mark 2}{(12.347,5.184)}
\gppoint{gp mark 2}{(12.347,5.903)}
\gppoint{gp mark 2}{(12.347,5.909)}
\gppoint{gp mark 2}{(12.347,5.926)}
\gppoint{gp mark 2}{(9.952,3.708)}
\gppoint{gp mark 2}{(9.870,3.762)}
\gppoint{gp mark 2}{(9.358,4.253)}
\gppoint{gp mark 2}{(10.093,4.253)}
\gppoint{gp mark 2}{(12.347,5.874)}
\gppoint{gp mark 2}{(12.347,5.943)}
\gppoint{gp mark 2}{(12.347,5.954)}
\gppoint{gp mark 2}{(12.347,5.959)}
\gppoint{gp mark 2}{(6.085,8.345)}
\gpcolor{gp lt color border}
\node[gp node right] at (7.737,8.345) {blimc};
\gpcolor{gp lt color 3}
\gppoint{gp mark 3}{(4.349,3.998)}
\gppoint{gp mark 3}{(12.347,4.078)}
\gppoint{gp mark 3}{(12.347,5.068)}
\gppoint{gp mark 3}{(12.347,5.172)}
\gppoint{gp mark 3}{(12.347,5.208)}
\gppoint{gp mark 3}{(12.347,4.836)}
\gppoint{gp mark 3}{(12.347,4.981)}
\gppoint{gp mark 3}{(12.347,5.741)}
\gppoint{gp mark 3}{(12.347,5.767)}
\gppoint{gp mark 3}{(12.347,5.991)}
\gppoint{gp mark 3}{(7.341,3.587)}
\gppoint{gp mark 3}{(7.527,3.450)}
\gppoint{gp mark 3}{(8.139,3.998)}
\gppoint{gp mark 3}{(7.841,4.078)}
\gppoint{gp mark 3}{(12.347,5.720)}
\gppoint{gp mark 3}{(12.347,5.767)}
\gppoint{gp mark 3}{(12.347,5.781)}
\gppoint{gp mark 3}{(12.347,5.787)}
\gppoint{gp mark 3}{(8.469,8.345)}
\gpcolor{gp lt color border}
\node[gp node right] at (10.121,8.345) {sp};
\gpcolor{gp lt color 4}
\gppoint{gp mark 4}{(1.688,0.985)}
\gppoint{gp mark 4}{(11.728,1.651)}
\gppoint{gp mark 4}{(4.776,2.532)}
\gppoint{gp mark 4}{(5.981,2.707)}
\gppoint{gp mark 4}{(8.139,2.707)}
\gppoint{gp mark 4}{(4.578,2.532)}
\gppoint{gp mark 4}{(12.347,2.532)}
\gppoint{gp mark 4}{(10.866,3.373)}
\gppoint{gp mark 4}{(12.347,3.450)}
\gppoint{gp mark 4}{(12.347,3.521)}
\gppoint{gp mark 4}{(2.717,0.985)}
\gppoint{gp mark 4}{(4.776,0.985)}
\gppoint{gp mark 4}{(2.717,1.651)}
\gppoint{gp mark 4}{(5.378,1.651)}
\gppoint{gp mark 4}{(4.776,3.521)}
\gppoint{gp mark 4}{(12.347,3.587)}
\gppoint{gp mark 4}{(12.347,3.587)}
\gppoint{gp mark 4}{(12.347,3.587)}
\gppoint{gp mark 4}{(10.853,8.345)}
\gpcolor{gp lt color border}
\gpsetlinetype{gp lt border}
\draw[gp path] (1.688,8.023)--(1.688,0.985)--(12.347,0.985);
\gpdefrectangularnode{gp plot 1}{\pgfpoint{1.688cm}{0.985cm}}{\pgfpoint{12.347cm}{8.023cm}}
\end{tikzpicture}